\documentclass[acmsmall,screen]{acmart}
\pdfoutput=1
\usepackage[utf8]{inputenc}
\usepackage{anyfontsize}
\usepackage{opendeductionmod}
\usepackage{stmaryrd}

\DeclareFontFamily{U}{mathx}{\hyphenchar\font45}
\DeclareFontShape{U}{mathx}{m}{n}{
      <5> <6> <7> <8> <9> <10>
      <10.95> <12> <14.4> <17.28> <20.74> <24.88>
      mathx10
      }{}
\DeclareSymbolFont{mathx}{U}{mathx}{m}{n}
\DeclareFontSubstitution{U}{mathx}{m}{n}
\DeclareMathAccent{\widehat}{0}{mathx}{"70}
\DeclareMathAccent{\widecheck}{0}{mathx}{"71}

\begin{document}

\theoremstyle{acmdefinition}\newtheorem{remark}[theorem]{Remark}

\title{A Subatomic Proof System for Decision Trees}
\author{Chris Barrett}
\orcid{0000-0003-1708-3554}
\affiliation{\institution{University of Bath}
                    \city{Bath}
                 \country{UK}}
\author{Alessio Guglielmi}
\orcid{0000-0002-7234-2347}
\affiliation{\institution{University of Bath}
                    \city{Bath}
                 \country{UK}}


\setcopyright{acmlicensed}
\acmJournal{TOCL}
\acmYear{2022} \acmVolume{1} \acmNumber{1} \acmArticle{1} \acmMonth{1} \acmPrice{15.00}\acmDOI{10.1145/3545116}


\begin{abstract}
We design a proof system for propositional classical logic that integrates two languages for Boolean functions: standard conjunction-disjunction-negation and binary decision trees. We give two reasons to do so. The first is proof-theoretical naturalness: the system consists of all and only the inference rules generated by the single, simple, linear scheme of the recently introduced subatomic logic. Thanks to this regularity, cuts are eliminated via a natural construction. The second reason is that the system generates efficient proofs. Indeed, we show that a certain class of tautologies due to Statman, which cannot have better than exponential cut-free proofs in the sequent calculus, have polynomial cut-free proofs in our system. We achieve this by using the same construction that we use for cut elimination. In summary, by expanding the language of propositional logic, we make its proof theory more regular and generate more proofs, some of which are very efficient.

That design is made possible by considering atoms as superpositions of their truth values, which are connected by self-dual, non-commutative connectives. A proof can then be projected via each atom into two proofs, one for each truth value, without a need for cuts. Those projections are semantically natural and are at the heart of the constructions in this paper. To accommodate self-dual non-commutativity, we compose proofs in deep inference.
\end{abstract}

\maketitle

\section{Introduction}

In this work, we contribute to the proof theory of classical logic in two ways: \begin{itemize}
\item we design a new proof system with a more natural and simpler normalisation theory;
\item we show that, within this proof system, significantly smaller cut-free proofs become available.
\end{itemize}
We do so by augmenting the standard, Boolean algebraic language of propositional classical logic formulae with decision trees \cite{Wege:00:Branchin:zr}. In other words, we build a new proof system around a language based on standard connectives (conjunction, disjunction, negation) and augmented by expressions in which atoms are connectives denoting choice: the formula $A\mathbin bC$ has the truth value of formula $A$ if atom $b$ is negative and that of $C$ if $b$ is positive. For example, if $c$ is negative, $A\land(B\mathbin c((D\lor E)\mathbin fG))$ is interpreted as $A\land B$. If $c$ is positive and $f$ negative, the formula is interpreted as $A\land(D\lor E)$. We limit ourselves, in this paper, to propositional proof systems, but those can be extended with quantifiers and modalities in standard ways. At first sight, the integration of the two languages of Boolean operators and decision trees is not compelling because each language is complete for Boolean functions. However, an important reason for that integration comes from normalisation. Developing the basic theory requires some work, which we do in this paper. The highlight is a very simple and natural cut-elimination procedure that is, to the best of our knowledge, based on new principles that could have very wide applicability in other parts of proof theory.

The technical starting point is the observation that extending the standard language by decision trees is natural from the point of view of \emph{subatomic logic} \cite{AlerGugl:17:Subatomi:dn} in deep inference \cite{Gugl:06:A-System:kl,BrunTiu:01:A-Local-:mz,GuglStra:01:Non-comm:rp}. In subatomic logic, several inference rules that seem unrelated in the standard proof systems, such as identity, cut, contraction and weakening, are obtained as instances of a more general scheme that generates them all. Not only the design of the rules is unified this way, but also the study of normalisation can be done uniformly. In subatomic logic, all inference rules are linear in the usual, structural sense of dealing with the same variables in the premiss and the conclusion, with no repetitions. Non-linearity, which can manifest itself in inferences as duplication or negation of atoms, is revealed through an interpretation map that does not affect normalisation and that maps the subatomic proofs to the standard ones. Surprisingly, the subatomic logic scheme that generates all the rules of propositional logic also generates sound rules for decision trees.

\newcommand{\ideq}{\mathrel\equiv}
Despite the novelty, and perhaps the exoticism, of the subatomic perspective, the constructions in this paper are very natural and can be easily understood, at least at an intuitive level, without any previous exposure to background material. Indeed, they could be considered a good introduction to the whole area of deep inference. Our aim for the rest of the introduction is to show the main ideas with a simple example. Let us consider the series of tautologies defined by Statman in \cite{Stat:78:Bounds-f:fj} (of which a more accessible, equivalent definition can be found in \cite{ClotKran:02:Boolean-:xr}, Section~5.3.2). The series is defined in Section~\ref{SectStatman}, but, for this introduction, the first three tautologies are sufficient:
\[
\renewcommand\arraystretch{1.5}
\begin{array}{@{}r@{}c@{}l@{}}
S_1&{}\ideq{}&\phantom{((}(\bar a_1\land\bar b_1)\lor(a_1\lor b_1)\;,       \\
S_2&{}\ideq{}&\lower1pc
              \hbox{\!\!\!$\renewcommand\arraystretch{1}
                           \begin{array}{l@{}l@{}l@{}l@{}}
                                   & &(\bar a_2\land\bar b_2)\lor{}       \\
                                  (&(&(a_2\lor b_2)\land\bar a_1)\land{}& \\
                                   &(&(a_2\lor b_2)\land\bar b_1)       &)
                                                         \lor(a_1\lor b_1)
                                                                       \;,\\
                           \end{array}$}                                    \\
S_3&{}\ideq{}&\lower2pc
              \hbox{\!\!\!$\renewcommand\arraystretch{1}
                           \begin{array}{l@{}l@{}l@{}l@{}}
                                   & &(\bar a_3\land\bar b_3)\lor{}       \\
                                  (&(&(a_3\lor b_3)\land\bar a_2)\land{}& \\
                                   &(&(a_3\lor b_3)\land\bar b_2)       &)
                                                                    \lor{}\\
                                  (&(&(a_3\lor b_3)\land
                                      (a_2\lor b_2)\land\bar a_1)\land{}& \\
                                   &(&(a_3\lor b_3)\land
                                      (a_2\lor b_2)\land\bar b_1)       &)
                                                         \lor(a_1\lor b_1)
                                                                       \;.\\
                          \end{array}$}                                     \\
\end{array}
\]
$S_1$ is obviously a tautology. To see that $S_2$ is a tautology, one can do a case analysis on the truth values of $a_2$ and $b_2$: if both atoms are false the tautology is trivial and, if any of them is true, $S_2$ reduces to $S_1$ because $a_2\lor b_2$ is true. The same argument works for all the other tautologies in the sequence $S_1$, \dots, $S_n$, \dots. That constitutes a polynomial-time proof method, and it is immediate to translate it into a formal proof of some proof system containing a cut rule, which can be used for the case analysis. Indeed, Statman proved that there are polynomial-size tree-like sequent proofs with cut for the tautologies, but only exponential-size tree-like ones if cut is not allowed. It would seem that if we remove the cut from a proof system, we remove the means for generating small proofs by case analysis. However, it turns out that augmenting the language with decision trees circumvents this problem: we can indeed obtain polynomial-size cut-free proofs of Statman tautologies based on case analysis.

\newcommand{\drB}{\drv[border=black]}
\newcommand{\lpr}[1]{\mathop{\strut\mathsf l_{#1}}}
\newcommand{\rpr}[1]{\mathop{\strut\mathsf r_{#1}}}
\catcode`@=11 
\def\finupsm@sh{\ht\z@\z@ \box\z@}%
\def\mathupsm@sh#1#2{\setbox\z@\hbox{$\m@th#1{#2}$}\finupsm@sh}%
\def\upsmash{\relax 
  \ifmmode\def\next{\mathpalette\mathupsm@sh}\else\let\next\makeupsm@sh
  \fi\next}%
\catcode`@=12 
We can  formalise the previous argument with decision trees in deep inference as follows. If we replace each occurrence of $a_2$ in $S_2$ with the unit $1$ (for `true'), we obtain a formula $\rpr{a_2}S_2$, which we call the \emph{right projection} of $S_2$. Similarly, in the \emph{left projection} $\lpr{a_2}S_2$, we substitute each occurrence of $a_2$ by $0$ (`false') and we obtain a formula $\lpr{a_2}S_2$. From $(\lpr{a_2}S_2)\mathbin{a_2}(\rpr{a_2}S_2)$ we can derive $S_2$ with a simple cut-free derivation (Lemma~\ref{LemmaReorder}). Therefore, we can represent the above argument about deriving $S_2$ from $S_1$ as
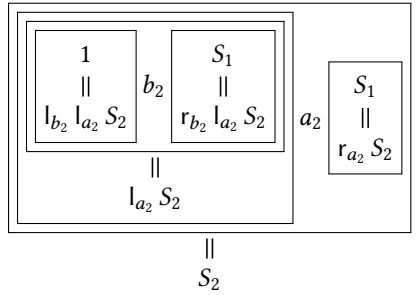
\[
\drv{\drB{\drB{\drB{\drB{1;
                         |;
                         \lpr{b_2}\lpr{a_2}S_2}
                    \mathbin{b_2}
                    \drB{S_1;
                         |;
                         \rpr{b_2}\lpr{a_2}S_2}};
               |;
               \lpr{a_2}S_2}
          \mathbin{a_2}
          \drB{S_1;
               |;
               \rpr{a_2}S_2}};
     |;
     S_2}
\;,
\]
where $\upsmash{\drv{A;|;B}}$ denotes a derivation from $A$ to $B$. The construction is shown in detail in Figures~\ref{FigStatmanBase} and \ref{FigStatmanInduct}.

In deep inference, particularly in subatomic logic, projections can be lifted from formulae to proofs with minimal effort. This way, we can transform a proof of $A$, which might contain cuts, into a proof of $(\lpr bA)\mathbin b(\rpr bA)$, where $b$ is an atom that we are free to choose. From $(\lpr bA)\mathbin b(\rpr bA)$, in turn, we derive $A$. The proof of $A$ so obtained has the remarkable property of not containing any cut on atom $b$. By repeating the construction on all the atoms appearing in cuts, we obtain a cut-free proof of $A$. It is interesting to note that, if we interpret the cuts in the original proof of $A$ as case analyses, we have turned them into decision points in a decision tree. Semantically, they might constitute the same case analysis, but they have an entirely different computational and complexity-theoretic behaviour. We think that this idea is general, novel and natural, and indeed the main interest of this paper because it might have far-reaching consequences in proof theory and proof complexity.

As expected, augmenting a language with new connectives makes new proofs available, and, among them, some are smaller than those in the original language. In this sense, the results in this paper are not a surprise. However, one would not expect the proof system for the augmented language to be less complex than the one for the original language. In fact, new connectives – in our case, atoms – need new inference rules, and this normally means that the amount of information necessary to describe the proof system increases. Surprisingly, in our case, the opposite happens: we can describe the proof system of the augmented language with less information than is necessary for the proof system of the language without atoms-as-connectives. The reason is that \emph{all} the instances of a general, but simple, proof scheme are sound for the extended language, and their collection is complete. This allows us to develop the proof theory of the new proof system working directly with one and only one subatomic scheme, and that is the main reason why we consider our proof system natural.

The findings in this paper corroborate those in \cite{AlerGugl:17:Subatomi:dn}, the foundational paper of subatomic logic, but the results in this paper have been obtained fortuitously, without any planning or expectation. What are the reasons for the unexpected success of the subatomic scheme? Unfortunately, we cannot offer any convincing explanation at this time. We hope that further research will allow us to understand these phenomena and formulate reasonable research questions.

\section{Decision Trees and the Subatomic Language}

A decision tree is a data structure that forms the basis of many modern, efficient implementations of Boolean functions. For example, the decision tree represented by $A\mathbin aB$ is read as `if $a$ then $B$, else $A$', where $a$ is an atom, \emph{i.e.}, a Boolean variable. Just as in subatomic logic, we consider atoms as binary self-dual, non-commutative connectives.

Suppose we wish to prove that two decision trees represent the same function; we could do so by building a proof system around decision trees. In what follows, we take an approach that enriches the language of decision trees with propositional connectives $\land$ and $\lor$ in order to express implication between formulae and gain access to the proof compression mechanisms of contraction and cut.

\begin{definition}
A \emph{decision tree} is a binary tree, where each interior node is labelled by an atom and each leaf is labelled by a unit ($0$ or $1$).
\end{definition}

We use the word `atom' both at the standard and the subatomic level. Typographically, we distinguish atoms in subatomic logic by using boldface.

\newcommand{\ba}{\mathbin{\mathbf a}}
\newcommand{\bb}{\mathbin{\mathbf b}}
\newcommand{\bc}{\mathbin{\mathbf c}}
\newcommand{\DT}{{\mathsf{DT}}}
\newcommand{\atoms}{{\mathcal A}}
\begin{definition}
Let $\atoms$ be a countable set of \emph{atoms} whose elements are denoted by $\ba$, $\bb$, $\bc$, \dots. We define the \emph{formulae} of the (subatomic) language $\DT$ as the language generated by the grammar
\[
\DT\mathrel{::=}0\mid1\mid(\DT\lor\DT)\mid(\DT\land\DT)\mid(\DT\mathbin\atoms\DT)
\;,
\]
where the distinct symbols $0$, $1$, $\lor$ and $\land$ do not appear in $\atoms$. We call $0$ and $1$ \emph{units}. Atoms, $\lor$ and $\land$ are \emph{connectives}. We denote formulae in $\DT$ as $A$, $B$, $C$, $D$, \dots. The expression $A\ideq B$ denotes syntactic equality of formulae. A \emph{context} is a formula where some subformulae are substituted by holes, which are denoted $\{\;\}$. Given a context with a single hole, denoted $K\{\;\}$, and a formula $A$, let $K\{A\}$ denote the formula given by substituting $A$ into the hole. In a formula $K\{A\ba B\}$, for any $\ba\in\atoms$, each atom occurrence in $A$ or $B$ is said to be \emph{nested}. We may drop parentheses when there is no ambiguity. The \emph{size} of a formula is the number of occurrences of units appearing in it.
\end{definition}

\begin{definition}
We equip the language $\DT$ with an involution $\bar\cdot\,\colon\DT\to\DT$, which we call \emph{negation}, defined inductively as follows, for all atoms $\ba$:
\[
\bar0=1
\;,\quad
\bar1=0
\;,\quad
\overline{A\lor B}=\bar A\land\bar B
\;,\quad
\overline{A\land B}=\bar A\lor\bar B
\;,\quad
\overline{A\ba B}=\bar A\ba\bar B
\;.
\]
We also use the same notation to denote the \emph{dual} of connectives, \emph{i.e.}, $\bar \lor=\land$, $\bar \land=\lor$ and $\bar \ba=\ba$.
\end{definition}

We take as equalities on formulae the linear unit equations of classical logic, plus two new subatomic unit equations.

\begin{definition}\label{DefEquiv}
The equivalence $=$ on formulae is the minimal equivalence relation, closed under context (\emph{i.e.}, $A=B$ implies $K\{A\}=K\{B\}$, for any context $K\{\;\}$), defined by the following equations:
\[
\renewcommand\arraystretch{1.5}
\begin{array}{@{}l@{\quad}l@{\quad}l@{}}
A\lor 0=A\;,&0\land0=0\;,&0\ba0=0\;,\\
A\land1=A\;,&1\lor 1=1\;,&1\ba1=1\;,\\
\end{array}
\]
for every atom $\ba$.
\end{definition}

For example, $1\lor(0\ba1)$ is equal under $=$ to $1\lor(((0\ba1)\land1)\lor(0\bb0))$, but neither formula is equal to $1$.

The language $\DT$ expresses both the language of decision trees and the language of classical propositional formulae. The latter is obtained via an interpretation function that interprets as standard formulae all the subatomic formulae that do not contain irreducible nested atoms. We characterise these two sublanguages as follows.

\newcommand{\Prop}{{\mathsf{Prop}}}
\newcommand{\SDT}{{\mathsf{SDT}}}
\begin{definition}\label{DefPropSDT}
Given $A\in\DT$, we say that $A$ is a \emph{propositional formula} if there is a $B=A$ with no nested atom occurrences; we say that $A$ is a \emph{strict decision tree} if there is a $B=A$ with no conjunctions or disjunctions. We denote with $\Prop$ and $\SDT$, respectively, the subsets of $\DT$ containing all the propositional formulae and the strict decision trees.
\end{definition}

For example, a formula that contains no nested atom occurrences, and thus is in $\Prop$, is $(1\ba0)\lor(1\bb0)$. Any formula equivalent to this is also in $\Prop$, even if it contains nested atom occurrences. For example, $(1\ba(0\bc0))\lor(1\bb0)$. Similarly, the formula $(1\bb0)\ba(1\bc1)$ is in $\SDT$, and thus so is the equivalent formula $(1\bb0)\ba(0\lor1)$, even though it contains a disjunction.

\newcommand{\p}{\mathop{\mathsf p\strut}}
\begin{remark}\label{RemEmbed}
There is a natural map $\p$ from the formulae of classical propositional logic to the language $\Prop$, given by substituting positive and negative atoms $a$ and $\bar a$ with the corresponding decision trees $0\ba1$ and $1\ba0$, respectively. For example, $\p((1\lor a)\land\bar b)=(1\lor(0\ba1))\land(1\bb0)$.
\end{remark}

\newcommand{\sem}[1]{\llbracket#1\rrbracket}
Throughout the paper, we exhibit constructions that allow us to translate between $\Prop$ and $\SDT$ within our proof system; that is, for any formula in $\Prop$, we can find a semantically equivalent formula in $\SDT$ and vice versa, and proofs, within the proof system, supporting that equivalence. $\Prop$ and $\SDT$ are syntax for Boolean functions and can be interpreted as such via a standard evaluation function $\sem{\,\cdot\,}$. In the following, we use the symbols $\land$ and $\lor$ also as Boolean operators on the set of values $\{0,1\}$, \emph{i.e.}, we use them in the semantics without typographically distinguishing them.

\begin{definition}\label{DefSem}
For any assignment of Boolean values to atoms $X\colon\atoms\to\{0,1\}$, we define the \emph{semantic interpretation} of formulae $\sem{\,\cdot\,}_X\colon\DT\to\{0,1\}$ as
\[
\renewcommand\arraystretch{1.5}
\begin{array}{@{}r@{}}
\sem{0}_X=0\;,\\
\sem{1}_X=1\;,\\
\end{array}
\qquad
\begin{array}{@{}r@{}}
\sem{A\lor  B}_X=\sem{A}_X\lor \sem{B}_X\;,\\
\sem{A\land B}_X=\sem{A}_X\land\sem{B}_X\;,\\
\end{array}
\qquad
\sem{A\ba B}_X =
\begin{cases}
\sem{A}_X&\text{if }X(\ba)=0\;,\\
\sem{B}_X&\text{if }X(\ba)=1\;,\\
\end{cases}
\]
for every atom $\ba$. We say that $A$ \emph{implies} $B$ if, for every assignment $X$, $\sem{A}_X\leq\sem{B}_X$. We say that $A$ is a \emph{tautology} if for every assignment $X$, $\sem{A}_X=1$. We say that two formulae $A$ and $B$ are \emph{semantically equivalent} if for every assignment $X$, $\sem{A}_X=\sem{B}_X$.
\end{definition}

For the sublanguages $\SDT$ and $\Prop$, the semantic interpretation coincides with the standard one.

\section{Deep Inference and a Subatomic Proof System for Decision Trees}

Deep inference is, essentially, the ability to apply inference rules at arbitrary depth within a formula \cite{AlerGugl:17:Subatomi:dn,Gugl:06:A-System:kl,GuglGundPari::A-Proof-:fk}. Thus, derivations are composable horizontally by the same connectives as formulae, as well as in the usual vertical manner. Subatomic proof systems \cite{Aler:16:A-Study-:hc,AlerGugl:17:Subatomi:dn} represent atoms as non-commutative self-dual connectives and, to accommodate this choice, we compose proofs by deep inference. Consequently, it is possible to adopt a unique, linear rule scheme that can generate all the standard inference rules for many logics, including classical, $\mathsf{MALL}$ \cite{Gira:87:Linear-L:wm} and $\mathsf{BV}$ \cite{Gugl:06:A-System:kl}.

\let\sup\widehat
\let\sdn\widecheck
\newcommand{\SKS}{{\mathsf{SKS}}}
For example, in system $\SKS$ for classical logic \cite{BrunTiu:01:A-Local-:mz}, non-linear rules such as identity and contraction appear in the subatomic system as the following instances of more general \emph{linear} inferences. Working modulo some unit equations, we can interpret the left sides as the respective right sides, and vice versa:
\[
\renewcommand\arraystretch{3}
\begin{array}{@{}r@{\quad\leftrightarrow\quad}l@{}}
\drv{(0\lor1)\ba(1\lor0);
     [\lor\sdn\ba]-;
     (0\ba1)\lor(1\ba0)}
&\drv{1;-;a\lor\bar a}\;,\\
\drv{(0\ba1)\lor(0\ba1);
     [\lor\sup\ba]-;
     (0\lor0)\ba(1\lor1)}
&\drv{a\lor a;-;a}\;.    \\
\end{array}
\]

Indeed, it turns out that \emph{all} the instances of the general rule scheme are sound and complete for our extended language. That is not the case for the standard language of propositional classical logic, to which two instances of the scheme do not apply because they have nested atoms. Thus, by extending the language, we simplify the proof system. This is so if we consider, as is customary in proof complexity, a proof system as an algorithm that checks proofs. A program applying one rule scheme is shorter than one that has to consider several, but not all, instances of the scheme. That comparison also holds in terms of Kolmogorov complexity.

The rule scheme is as follows.

\newcommand{\salp}{\mathbin\alpha}
\newcommand{\sbet}{\mathbin\beta}
\begin{definition}
The inference rules of \emph{subatomic shape} are those of the form
\[
\drv{(A\sbet B)\salp(C\mathbin{\sup\sbet}D);
     [\alpha\sup\beta]-;
     (A\salp C)\sbet(B\salp D)}
\quad\text{or}\quad
\drv{(A\sbet B)\salp(C\sbet D);
     [\beta\sdn\alpha]-;
     (A\salp C)\sbet(B\mathbin{\sdn\salp}D)}\;,
\]
for connectives $\alpha,\beta\in\{{\lor},{\land}\}\cup\atoms$ and formulae $A$, $B$, $C$ and $D$ of $\DT$. The \emph{names} of inference rules of subatomic shape are given by $\alpha\sup\beta$ and $\beta\sdn\alpha$ respectively. We define $\sdn\land=\sdn\lor=\lor$, $\sup\land=\sup\lor=\land$ and $\sdn\ba=\sup\ba=\ba$. This notation relates two dual connectives, assigning a \emph{weaker} $\sdn\cdot$ and a \emph{stronger} $\sup\cdot$ of the pair.
\end{definition}

\begin{remark}
Note that the name $\ba\sdn\land$ is not equal to $\ba\lor$ (which is not considered a name at all), even though $\sdn\land=\lor$. That is to say, although $\sup\cdot$ and $\sdn\cdot$ are considered as functions on connectives when appearing in a rule instance, they are considered as decorations when appearing in the name of a rule.
\end{remark}

\begin{remark}
Some subatomic rules are assigned two names under this scheme, for example $\lor\sup\land$ and $\land\sdn\lor$, as in
\[
\drv{(A\land B)\lor(C\mathbin{\sup\land}D);
     [\lor\sup\land]-;
     (A\lor C)\land(B\lor D)}
\quad\text{and}\quad
\drv{(A\land B)\lor(C\land D);
     [\land\sdn\lor]-;
     (A\lor C)\land(B\mathbin{\sdn\lor}D)}
\;.
\]
We consider this one rule with two distinct names. Another example is the pair $\ba\sup\bb$ and $\ba\sdn\bb$:
\[
\drv{(A\bb B)\ba(C\bb D);
      [\ba\sup\bb]-;
      (A\ba C)\bb(B\ba D)}
\quad\text{and}\quad
\drv{(A\ba B)\bb(C\ba D);
      [\ba\sdn\bb]-;
      (A\bb C)\ba(B\bb D)}
\]
(where one should keep in mind that $\ba$ and $\bb$ stand for any atom).
\end{remark}

In the deep-inference literature, the cut rule is typically decomposed into several `up' rules and the identity into several `down' rules. It would be tempting to associate the saturation function with the up-down classification of rules. However, we do not have, at this point, any strong reason to do so, therefore we proceed in an agnostic way and we advise the reader not to make any assumption in this regard.

\newcommand{\DTsa}{{\DT^{\mathsf{sa}}}}
We define our subatomic proof system for decision trees, which we call $\DTsa$, as the set of all the possible rules of subatomic shape.

\begin{definition}
A \emph{subatomic proof system} is a set of inference rules of subatomic shape, together with \emph{equality rules}, for an equational theory $=$, given by $\drv{A;[=]-;B}$, for every formulae $A$ and $B$ such that $A=B$. System $\DTsa$ is the subatomic proof system obtained by taking all the rules of subatomic shape $\alpha\sdn{\beta}$ and $\alpha\sup\beta$ for connectives $\alpha,\beta\in\{{\lor},{\land}\}\cup\atoms$ together with the set of equality rules generated by $=$ from Definition~\ref{DefEquiv}.
\end{definition}

\begin{figure}[t]
\[
\renewcommand\arraystretch{2.5}
\begin{array}{@{}c@{\qquad}c@{\qquad}c@{}}
 \drv{(A\land B)\land(C\land D);
      [\land\sup\land]-;
      (A\land C)\land(B\land D)}
&\drv{(A\land B)\lor(C\land D);
      [\lor\sup\land]-;
      (A\lor C)\land(B\lor D)}
&\drv{(A\land B)\ba(C\land D);
      [\ba\sup\land]-;
      (A\ba C)\land(B\ba D)}    \\
 \drv{(A\lor B)\land(C\land D);
      [\land\sup\lor]-;
      (A\land C)\lor(B\land D)}
&\drv{(A\lor B)\lor(C\land D);
      [\lor\sup\lor]-;
      (A\lor C)\lor(B\lor D)}
&\drv{(A\lor B)\ba(C\land D);
      [\ba\sup\lor]-;
      (A\ba C)\lor(B\ba D)}     \\
 \drv{(A\ba B)\land(C\ba D);
      [\land\sup\ba]-;
      (A\land C)\ba(B\land D)}
&\drv{(A\ba B)\lor(C\ba D);
      [\lor\sup\ba]-;
      (A\lor C)\ba(B\lor D)}
&\drv{(A\bb B)\ba(C\bb D);
      [\ba\sup\bb]-;
      (A\ba C)\bb(B\ba D)}      \\\noalign{\bigskip}
 \drv{(A\lor B)\lor(C\lor D);
      [\lor\sdn\lor]-;
      (A\lor C)\lor(B\lor D)}
&\drv{(A\land B)\lor(C\land D);
      [\land\sdn\lor]-;
      (A\lor C)\land(B\lor D)}
&\drv{(A\ba B)\lor(C\ba D);
      [\ba\sdn\lor]-;
      (A\lor C)\ba(B\lor D)}    \\
 \drv{(A\lor B)\land(C\lor D);
      [\lor\sdn\land]-;
      (A\land C)\lor(B\lor D)}
&\drv{(A\land B)\land(C\land D);
      [\land\sdn\land]-;
      (A\land C)\land(B\lor D)}
&\drv{(A\ba B)\land(C\ba D);
      [\ba\sdn\land]-;
      (A\land C)\ba(B\lor D)}   \\
 \drv{(A\lor B)\ba(C\lor D);
      [\lor\sdn\ba]-;
      (A\ba C)\lor(B\ba D)}
&\drv{(A\land B)\ba(C\land D);
      [\land\sdn\ba]-;
      (A\ba C)\land(B\ba D)}
&\drv{(A\ba B)\bb(C\ba D);
      [\ba\sdn\bb]-;
      (A\bb C)\ba(B\bb D)}      \\
\end{array}
\]
\caption{System $\DTsa$ (excluding equations). Every inference rule is generated by the unique subatomic rule scheme, and no instances of the scheme are excepted from inclusion. Atoms $\ba$ and $\bb$ range over all of $\atoms$. Some rules such as $\ba\sdn\bb$ and $\land\sdn\lor$ appear twice (denoted by distinct names).}
\Description{System $\DTsa$ (excluding equations). Every inference rule is generated by the unique subatomic rule scheme, and no instances of the scheme are excepted from inclusion. Atoms $\ba$ and $\bb$ range over all of $\atoms$. Some rules such as $\ba\sdn\bb$ and $\land\sdn\lor$ appear twice (denoted by distinct names).}
\label{FigSystem}
\end{figure}

See Figure~\ref{FigSystem} for the complete list of non-equality inference rules of $\DTsa$. The list of rules in the figure is only provided to facilitate the reading of the various derivations in this paper. We hope that the reader keeps in mind that all those rules are generated by the subatomic scheme and, indeed, the best way to understand our work is to familiarise oneself with the scheme rather than using Figure~\ref{FigSystem} as a dictionary.

\newcommand{\SKSsa}{{\SKS^{\mathsf{sa}}}}
The first subatomic system for classical logic was introduced in \cite{AlerGugl:17:Subatomi:dn} as $\mathsf{SAKS}$. In this paper, we refer to that system as $\SKSsa$. Note that we deal with commutativity and associativity of $\land$ and $\lor$ using inference rules of subatomic shape, rather than equality rules.

\begin{definition}
The subatomic proof system for classical logic, $\SKSsa$, is the set of inference rules $\{{\lor\sdn\land},{\land\sdn\lor},{\lor\sdn\lor},{\land\sup\lor},{\lor\sup\land},{\land\sup\land}\} \cup(\bigcup_{\ba\in\atoms}\{{\lor\sdn\ba},{\land\sdn\ba},{\land\sup\ba},{\lor\sup\ba}\})$, together with the set of equality rules generated by $=$ from Definition~\ref{DefEquiv}.
\end{definition}

\begin{figure}[t]
\[
\drv{1;
     [=]-;
     \left(
     1\lor\left(\,
          \drB{(0\land1)\lor(1\land0);
               [\land\sdn\lor]-;
               (0\lor1)\land(1\lor0);
               [=]-;
               1}
          \bb0
          \right)
     \right)
     \ba(0\lor(1\bb1));
     [\lor\sdn\ba]-;
     (1\ba0)\lor\drB{(1\bb0)\ba(1\bb1);
                     [\bb\sdn\ba]-;
                     (1\ba1)\bb(0\ba1)}}
\]
\caption{An example proof in system $\DTsa$.}
\Description{An example proof in system $\DTsa$.}
\label{FigSmallProof}
\end{figure}

The system $\DTsa$ can be seen as the subatomic system $\SKSsa$ for classical logic extended with the rule $\ba\sdn\bb$ (or, equivalently, $\ba\sup\bb$), which allows the manipulation of decision trees. For example, the proof in Figure~\ref{FigSmallProof} would be in $\SKSsa$ except for this rule. In other words, we can define $\SKSsa=\DTsa\setminus{\bigcup}_{\ba,\bb\in\atoms}\{{\ba\sdn\bb},{\ba\sup\bb}\}$.

Note that this definition includes some rules not originally in $\SKSsa$ as it appears in \cite{AlerGugl:17:Subatomi:dn}, namely $\{{\land\sdn\land},{\ba\sdn\land},{\lor\sup\lor},{\ba\sup\lor}\}$. However, each of these rules can be replaced with a derivation in the original system for only a linear cost in complexity. An example of this can be found after the following definition, which defines deep-inference proofs according to the open deduction formalism \cite{GuglGundPari::A-Proof-:fk}.

\newcommand{\wid}{\mathop{\mathsf w\strut}}
\newcommand{\hei}{\mathop{\mathsf h\strut}}
\newcommand{\pr}{\mathop{\mathsf{pr}\strut}}
\newcommand{\cn}{\mathop{\mathsf{cn}\strut}}
\newcommand{\proofsys}{{\mathcal S}}
\begin{definition}
Given a subatomic proof system $\proofsys$, a \emph{derivation} $\phi$ in $\proofsys$ with \emph{premiss} $A\ideq\pr\phi$ and \emph{conclusion} $B\ideq\cn\phi$, also said \emph{from} $A$ \emph{to} $B$, is denoted by
\[
\drB{A;[\phi]|[\proofsys];B}\;;
\]
its \emph{width} is denoted by $\wid\phi$ and its \emph{height} by $\hei\phi$. We define derivations and the functions $\pr$, $\cn$, $\wid$ and $\hei$ inductively; given $\psi$ and $\chi$ derivations in $\proofsys$, a derivation $\phi$ can be:
\begin{enumerate}
\item $\phi\ideq A$, where $A$ is a formula; in this case, $\pr\phi\ideq\cn\phi\ideq A$, $\wid\phi$ is the size of $A$ and $\hei\phi=0$;
\item a \emph{composition by inference}
\[
\phi
\ideq
\drB{\psi;
     [\rho]-;
     \chi}\;,
\]
where $\drv{\cn\psi;[\rho]-;\pr\chi}$ is an instance of an inference rule in $\proofsys$, which could be an equality rule; in this case, $\pr\phi\ideq\pr\psi$, $\cn\phi\ideq\cn\chi$, $\wid\phi=\max(\wid\psi,\wid\chi)$ and $\hei\phi=\hei\psi+\hei\chi+1$;
\item a \emph{composition by connective}
\[
\phi\ideq(\psi\salp\chi)\;,
\]
where $\alpha\in\{{\lor},{\land}\}\cup\atoms$; in this case, $\pr\phi\ideq(\pr\psi)\salp(\pr\chi)$, $\cn\phi\ideq(\cn\psi)\salp(\cn\chi)$, $\wid\phi=\wid\psi+\wid\chi$ and $\hei\phi=\max(\hei\psi,\hei\chi)$.
\end{enumerate}
We consider composition by inference associative. We call a derivation with premiss $1$ a \emph{proof}. The \emph{size} of a derivation is the number of occurrences of units appearing within it. We omit the name of a derivation or of a proof system if there is no ambiguity. The boxes around derivations are essentially brackets, and so we sometimes also omit them.
\end{definition}

\begin{remark}
Later in the text, we refer to the \emph{mirror image} of a derivation. For example, the following two rules are one the mirror image of the other:
\[
\drv{(A\lor B)\land(C\lor D);
     [\lor\sdn\land]-;
     (A\land C)\lor(B\lor D)}
\quad\text{and}\quad
\drv{(D\lor C)\land(B\lor A);
     -;
     (D\lor B)\lor(C\land A)}\;.
\]
Technically, the rule at the right is not included in the proof system. We can either choose to include the mirror image rules in the proof system, or for each mirror image construction we can make use of commutativity. In both cases, all important proof-theoretic properties are preserved.
\end{remark}

\begin{proposition}
System\/ $\DTsa$ is sound with respect to the semantics defined by\/ $\sem{\,\cdot\,}$, \emph{i.e.}, if there exists a derivation in\/ $\DTsa$ with premiss $A$ and conclusion $B$, then $A$ implies $B$.
\end{proposition}

\begin{proof}
By inspection of the rules.
\end{proof}

\begin{remark}
For every inference rule $\alpha\sdn\beta$ of subatomic shape from $A$ to $B$, there exists a \emph{dual inference rule} $\gamma\sup\delta$ of subatomic shape from $\bar B$ to $\bar A$, with $\gamma\ideq\bar\alpha$ and $\delta\ideq\bar\beta$, and vice versa. Thus, given a derivation $\phi$ from $A$ to $B$ in $\DTsa$, we can obtain a \emph{dual derivation} from $\bar B$ to $\bar A$ by inverting $\phi$, dualising every formula occurring in it and renaming each inference rule as its dual. For example, the dual of
\[
\drv{\drB{(A\lor B)\land(C\lor D);
          [\lor\sdn\land]-;
          (A\land C)\lor(B\lor D)}
     \land
     \drB{(E\lor F)\land(G\lor H);
          [\lor\sdn\land]-;
          (E\land G)\lor(F\lor H)};
     [\lor\sdn\land]-;
     ((A\land C)\land(E\land G))\lor
     ((B\lor  D)\lor (F\lor  H))}
\]
is
\[
\drv{((\bar A\lor \bar C)\lor (\bar E\lor \bar G))\land
     ((\bar B\land\bar D)\land(\bar F\land\bar H));
     [\land\sup\lor]-;
     \drB{(\bar A\lor \bar C)\land(\bar B\land\bar D);
          [\land\sup\lor]-;
          (\bar A\land\bar B)\lor (\bar C\land\bar D)}
     \lor
     \drB{(\bar E\lor \bar G)\land(\bar F\land\bar H);
          [\land\sup\lor]-;
          (\bar E\land\bar F)\lor (\bar G\land\bar H)}}\;.
\]
The dual of a sound derivation is sound.
\end{remark}

The following definition introduces derivation composition.

\begin{definition}
Let $\phi$ and $\psi$ be two derivations such that $\cn\phi\ideq\pr\psi$. We define
\[
\drB{\phi;
     .;
     \psi}\;,
\]
which we call the \emph{(synchronal) composition} of $\phi$ and $\psi$, as follows:
\begin{enumerate}\setlength{\itemsep}{1pc}
\item if $\phi\ideq A\in\DT$, then $\drB{A;.;\psi}\ideq\psi$; similarly, if $\psi\ideq A\in\DT$, then $\drB{\phi;.;A}\ideq\phi$;
\item if $
\phi\ideq\drB{\chi;[\rho]-;\omega}\;$, then $\drB{\drB{\chi;[\rho]-;\omega};.;\psi}\ideq\drB{\chi;[\rho]-;\drB{\omega;.;\psi}}\;$;
similarly, if $
\psi\ideq\drB{\chi;[\rho]-;\omega}\;$, then $\drB{\phi;.;\drB{\chi;[\rho]-;\omega}}\ideq\drB{\drB{\phi;.;\chi};[\rho]-;\omega}\;$;
\item if $\phi\ideq\chi\salp\omega$ and $\psi\ideq\chi'\salp\omega'$, then
$\drB{\chi\salp\omega;.;\chi'\salp\omega'}\ideq\drB{\chi;.;\chi'}\salp\drB{\omega;.;\omega'}\;$.
\end{enumerate}
\end{definition}

For example, we have
\[
\drB{\drB{(A\land B)\lor(C\land D);
          [\land\sdn\lor]-;
          (A\lor C)\land(B\lor D)}
     \land E;
     .;
     \drB{(A\lor C)\land(B\lor D);
          [\lor\sdn\land]-;
          (A\land B)\lor(C\lor D)}\land E}
\ideq
\drB{(A\land B)\lor(C\land D);
     [\land\sdn\lor]-;
     (A\lor C)\land(B\lor D);
     [\lor\sdn\land]-;
     (A\land B)\lor(C\lor D)}
\land E
\;.
\]

We call cuts those inferences that correspond to an atomic cut in system $\SKS$ via the embedding in Remark~\ref{RemEmbed} and the equational theory $=$.

\begin{definition}
A \emph{cut on} $\ba$ is any instance of the rule
\[
\drv{(A\ba B)\land(C\ba D);
     [\land\sup\ba]-;
     (A\land C)\ba(B\land D)}\;,
\]
such that $A=D=0$ and $B=C=1$, or $A=D=1$ and $B=C=0$ in the equational theory $=$. We define an \emph{identity} as the dual of a cut. A \emph{cut-free} (resp., \emph{identity-free}) derivation is a derivation that contains no cuts (resp., identities) on any atom.
\end{definition}

For example, a cut on $a$, corresponding to the standard atomic cut rule
\[
\drv{a\land\bar a;
-;
0}
\]
is
\[
\drv{((0\land0)\ba  (0\lor1))\land((1\lor1)\ba  (0\land1));
     [\land\sup\ba]-;
     ((0\land0)\land(1\lor1))\ba  ((0\lor1)\land(0\land1))}\;.
\]

At the subatomic level, what we call a cut does not look at all like a cut in the traditional sense. However, a subatomic cut is the only way to produce a standard cut via the embedding. Therefore, if we eliminate the subatomic cuts, we also eliminate the standard ones.

Indeed, the whole normalisation theory of $\SKSsa$, when restricted to $\Prop$, can be lifted from that of $\SKS$, which is the standard system for the standard language of propositional logic. In particular, we have the following proposition.

\begin{proposition}\label{PropImplComplCutFree}
$\SKSsa$ is implicationally complete for\/ $\Prop$ and cut-free complete for tautologies of\/ $\Prop$.
\end{proposition}

\begin{proof}
The proof follows from the following three facts (see \cite{AlerGugl:17:Subatomi:dn} for more details).
\begin{itemize}
\item If there is a derivation from $A$ to $B$ in $\SKS$, then there is a derivation from $\p A$ to $\p B$ in $\SKSsa$, where $\p$ is the map of Remark~\ref{RemEmbed}; the latter derivation is cut-free if the former is. Such a derivation is obtained by translating via $\p$ all its formulae and applying the inference rules of $\SKSsa$ in an obvious way.
\item $\SKS$ is implicationally complete and cut-free complete.
\item For every formula $C$ in $\Prop$ there is a formula $D$ in the standard language of propositional logic such that $C=\p D$.
\qedhere
\end{itemize}
\end{proof}

The implicational completeness and cut-free completeness of $\DTsa$ also rely on that of $\SKS$, in a very similar way to the above. We prove them in Propositions~\ref{PropImplCompl} and \ref{PropCutFreeCompl}, but, for those, we need to develop a bit of theory.

\section{Completeness}

In this section, we prove the completeness of our proof system with respect to the semantics of Definition~\ref{DefSem}. We first describe a few constructions that are useful in this proof, and later in the paper.

\begin{proposition}\label{PropWeak}
In system\/ $\DTsa$, for every formula $A$, there exist cut-free and identity-free derivations of the form
\[
\drv{0;
     |;
     A}
\quad\text{and}\quad
\drv{A;
     |;
     1}\;.
\]
Their width is $O(n)$ and their height is constant, where $n$ is the size of $A$.
\end{proposition}

\begin{proof}
We construct the first derivation by replacing each unit $1$ appearing in $A$ with
\[
\drv{0;
     [=]-;
     (0\land1)\lor(1\land0);
     [\land\sdn\lor]-;
     (0\lor1)\land(1\lor0);
     [=]-;
     1}\;.
\]
This makes the premiss equal to $0$ and the conclusion equal to $A$. The second derivation can be obtained as the dual of the first. No cuts or identities are used.
\end{proof}

\begin{definition}
We call \emph{weakenings} and \emph{coweakenings} respectively the derivations constructed in Proposition~\ref{PropWeak}.
\end{definition}

\begin{lemma}\label{LemmaConstr}
In system\/ $\DTsa$, for every formula $A$, there exist cut-free and identity-free derivations of the form
\[
\drv{A;
     |;
     A\ba1}
\;,\quad
\drv{A;
     |;
     1\ba A}
\quad\text{and}\quad
\drv{A\ba0;
     |;
     A}
\;,\quad
\drv{0\ba A;
     |;
     A}\;.
\]
Their height and width are $O(n)$, where $n$ is the size of $A$.
\end{lemma}

\begin{proof}
We construct the first derivation by structural induction on $A$. The second derivation can be constructed as the mirror image. The final two can be recovered as duals of the first two.

If $A\ideq0$, the necessary derivation is given by a weakening. If $A\ideq1$, take the derivation
\[
\drv{1;[=]-;1\ba1}\;.
\]

If $A\ideq B\salp C$ for some connective $\alpha$, construct
\[
\drv{\drB{B;
          |;
          B\ba1}
     \salp
     \drB{C;
          |;
          C\ba1};
     [\alpha\sup\ba]-;
     (B\salp C)\ba\drB{1\salp1;[=]-;1}}\;.
\]
Note that the occurrence of $\alpha\sup\ba$ here is not an instance of the cut rule, and, when dualised, not an instance of the identity rule.
\end{proof}

The semantic equivalence between $B\ba C$ and $(B\land(1\ba0))\lor((0\ba1)\land C)$ is cut-free provable within the system, \emph{i.e.}, there is a cut-free derivation from one to the other and vice versa. These constructions allow us to nest and un-nest atom occurrences, which in turn allows us to show that for any formula of $\DT$, there is a semantically equivalent formula in $\Prop$ (see Definition~\ref{DefPropSDT}), and vice versa, and that this equivalence is provable within the system. This reduces the completeness of system $\DTsa$ to that of system $\SKSsa$, which is known.

\begin{lemma}\label{LemmaNest}
In system\/ $\DTsa$, for every formulae $B$ and $C$, we can construct cut-free derivations of the form
\[
\drv{B\ba C;
     |;
     (B\land(1\ba0))\lor((0\ba1)\land C)}
\quad\text{and}\quad
\drv{(B\land(1\ba0))\lor((0\ba1)\land C);
     |;
     B\ba C}\;.
\]
Their width and height are $O(n)$, where $n$ is the size of $(B\ba C)$.
\end{lemma}

\begin{proof}
If $B \neq 0 \neq C$, we construct the two derivations as
\[
\drv{\drB{B;
          [=]-;
          (B\land1)\lor(0\land0)}
     \ba
     \drB{C;
          [=]-;
          (0\land0)\lor(1\land C)};
     [\lor\sdn\ba]-;
     \drB{(B\land1)\ba(0\land0);
         [\land\sdn\ba]-;
         \drB{B\ba0;[\phi]|;B}\land(1\ba0)}
     \lor
     \drB{(0\land0)\ba(1\land C);
          [\land\sdn\ba]-;
          (0\ba1)\land\drB{0\ba C;[\psi]|;C}}}
\]
and
\[
\drv{\drB{\drB{B;[\chi]|;B\ba1}
          \land
          (1\ba0);
          [\land\sup\ba]-;
          (B\land1)\ba(1\land0)}
     \lor
     \drB{(0\ba1)
          \land
          \drB{C;[\omega]|;1\ba C};
          [\land\sup\ba]-;
          (0\land1)\ba(1\land C)};
     [\lor\sup\ba]-;
     \drB{(B\land1)\lor(0\land 1);[=]-;B}
     \ba
     \drB{(1\land0)\lor(1\land C);[=]-;C}}\;,
\]
respectively, where $\phi$, $\psi$, $\chi$ and $\omega$ are instances of the constructions of Lemma~\ref{LemmaConstr}. No occurrence of rule $\land\sup\ba$ can be a cut. If $B=0$ or $C=0$, the latter derivation contains a cut in at least one of the subderivations
\[
\drB{\drB{B;[\chi]|;B\ba1}
     \land
     (1\ba0);
     [\land\sup\ba]-;
     (B\land1)\ba(1\land0)}
\quad\text{and}\quad
\drB{(0\ba1)
     \land
     \drB{C;[\omega]|;1\ba C};
     [\land\sup\ba]-;
     (0\land1)\ba(1\land C)}
\;.
\]
To fix this, we replace the problematic subderivation(s) with
\[
\drB{\drB{B;[=]-;0\ba0}
     \land(1\ba0);
     [\land\sup\ba]-;
     \left(\,\drB{0;[=]-;B}\land1\right)
     \ba
     \left(\,\drB{0;[\pi]|;1}\land0\right)}
\quad\text{and}\quad
\drB{(0\ba1)\land\drB{C;[=]-;0\ba0};
     [\land\sup\ba]-;
     \left(0\land\drB{0;[\theta]|;1}\,\right)
     \ba
     \left(1\land\drB{0;[=]-;C}\,\right)}
\;,
\]
respectively, where in each case $\pi$ and $\theta$ are instances of the weakening construction.

The width of these derivations is $O(n)$. The derivations $\phi$, $\psi$, $\chi$ and $\omega$ each have height $O(n)$ and thus each derivation has height $O(n)$, too.
\end{proof}

\begin{lemma}\label{LemmaEquiv}
In system\/ $\DTsa$, for every formula $A$ in\/ $\DT$, there exist some semantically equivalent formula $B$ in\/ $\Prop$ and cut-free derivations
\[
\drv{A;|;B}
\quad\text{and}\quad
\drv{B;|;A}
\;.
\]
\end{lemma}

\begin{proof}
We construct the first derivation. If there exist nested atoms occurring in $A$, then we have $A\ideq K\{C\ba D\}$ for some atom $\ba$, where either $C$ or $D$ contain an atom occurrence. Applying the first construction from Lemma~\ref{LemmaNest} to $C\ba D$ strictly reduces the number of nested atoms in $A$. Repeatedly applying this construction until there are no more nested atoms yields a derivation with conclusion $B$ that is semantically equivalent to $A$, and such that $B$ is in $\Prop$.

Repeatedly applying the second construction from Lemma~\ref{LemmaNest} to $B$ allows us to nest atom occurrences until we have recovered $A$.
\end{proof}

The following proposition derives from the same idea as Proposition~\ref{PropImplComplCutFree}: the implicational completeness for the extended language is derived from the implicational completeness of $\SKS$.

\begin{proposition}\label{PropImplCompl}
System\/ $\DTsa$ is implicationally complete with respect to the semantics\/ $\sem{\,\cdot\,}$, \emph{i.e.}, if $A$ implies $B$ then there exists a derivation with premiss $A$ and conclusion $B$.
\end{proposition}

\begin{proof}
Using Lemma~\ref{LemmaEquiv}, we can find formulae $C$ and $D$ in $\Prop$ that are semantically equivalent to $A$ and $B$ respectively, together with derivations $\phi$ from $A$ to $C$ and $\chi$ from $D$ to $B$ in $\DTsa$. Because $C$ and $D$ are both in $\Prop$, and because $C$ implies $D$, there exists a derivation with premiss $C$ and conclusion $D$ in system $\SKSsa$ by implicational completeness of $\SKSsa$, as per Proposition~\ref{PropImplComplCutFree}, which we call $\psi$. We can thus construct within $\DTsa$ the derivation
\[
\vbox{\hbox{%
\drv{A;[\phi]|[\DTsa ];C
     ;.;
     C;[\psi]|[\SKSsa];D
     ;.;
     D;[\chi]|[\DTsa ];B}\;.
}\kern0pt}\qedhere
\]
\end{proof}

The following proposition also derives from the same idea as Proposition~\ref{PropImplComplCutFree}. We do have here a cut-elimination procedure, but an indirect one, relying on a translation to $\Prop$.

\begin{proposition}\label{PropCutFreeCompl}
System\/ $\DTsa$ is cut-free complete with respect to the semantics\/ $\sem{\,\cdot\,}$, \emph{i.e.}, if $A$ is a tautology, there exists a cut-free proof of $A$.
\end{proposition}

\begin{proof}
Using Lemma~\ref{LemmaEquiv}, we can find a formula $B$ in $\Prop$ that is semantically equivalent to $A$, together with a cut-free derivation $\psi$ from $B$ to $A$. Since $B$ is a tautology in $\Prop$, there exists a cut-free proof $\phi$ of $B$ in $\SKSsa$ by the cut-free completeness of $\SKSsa$ for $\Prop$, as per Proposition~\ref{PropImplComplCutFree}. In fact, this is also a proof in $\DTsa$. Thus construct
\[
\vbox{\hbox{%
\drv{1;[\phi]|[\SKSsa];B
     ;.;
     B;[\psi]|[\DTsa ];A}\;.
}\kern0pt}\qedhere
\]
\end{proof}

Can we obtain a direct, and better, proof of cut elimination? The next section is devoted to that.

\section{Cut Elimination}

In this section, we describe a construction whereby a derivation can be projected, via a given atom, into two derivations, one for each truth value of that atom.

More in detail, we know that if we are given a formula, there is an obvious notion of left (resp., right) projection on an atom $\ba$ given by taking the value of that formula when $\ba$ is false (resp., true). This can be done both on standard and subatomic formulae. As we see in this section, in subatomic logic, we can lift the definition of projection from formulae to derivations with minimal effort, because of the linear nature of the inference scheme. Thus, we can transform a proof of $A$ into a proof of $(\lpr\ba A)\ba(\rpr\ba A)$, where $\ba$ is an atom that we are free to choose and $\lpr\ba A$ and $\rpr\ba A$ are left and right projections of $A$ on $\ba$, respectively (the precise definition follows in this section). In that proof, there are no cuts on $\ba$ because there are no occurrences of $\ba$ in the projections. Remarkably, we can also build a proof from $(\lpr\ba A)\ba(\rpr\ba A)$ to $A$ that does not contain cuts on $\ba$. Therefore, by repeating the construction on all the atoms appearing in cuts, we obtain a cut-free proof of $A$.

\begin{lemma}\label{LemmaContr}
In system\/ $\DTsa$, for every formula $A$ and every connective $\alpha\in\{\lor\}\cup\atoms$, $\beta\in\{\land\}\cup\atoms$, we can construct cut-free and identity-free derivations of the form
\[
\drv{A\salp A;|;A}
\quad\text{and}\quad
\drv{A;|;A\sbet A}
\;.
\]
Their size is $O(n^2)$, where $n$ is the size of $A$.
\end{lemma}

\begin{proof}
We construct the first derivation by structural induction on $A$. The second derivation can be recovered as the dual of the first.

\newcommand{\sgam}{\mathbin\gamma}
If $A\ideq0$, take $\drv{0\salp0;[=]-;0}$. If $A\ideq1$, take $\drv{1\salp1;[=]-;1}$. Otherwise, we have $A\ideq B\sgam C$, for some connective $\gamma$, and we construct the derivation
\[
\drv{(B\sgam C)\salp(B\sgam C);
     [\gamma\sdn\alpha]-;
     \drB{B\salp B;|;B}
     \sgam
     \drB{C\salp C;|;C}}\;.
\]
The width of the derivation is $O(n)$, and its height is $O(n)$.
\end{proof}

\begin{definition}
We call \emph{contractions} and \emph{cocontractions} respectively the derivations constructed in Lemma~\ref{LemmaContr}.
\end{definition}

\begin{lemma}\label{LemmaMerge}
In system\/ $\DTsa$, for every formulae $A$ and $B$ and every context $K\{\;\}$, there exist cut-free and identity-free derivations of the form
\[
\drv{K\{A\}\land B;
     |;
     K\{A\land B\}}
\quad\text{and}\quad
\drv{ K\{A\lor B\};
     |;
     K\{A\}\lor B}
\;.
\]
Their width is $O(m + n + l)$ and their height is $O(ml)$, where $m$ is the size of $K\{\;\}$, $n$ is the size of $A$ and $l$ is the size of $B$.
\end{lemma}

\begin{proof}
We construct the first derivation; the second can be obtained as the dual of the first. We proceed by structural induction on $K\{\;\}$. In the base case, $K\{\;\}\ideq\{\;\}$ and we take the trivial derivation. If $K\{\;\}\ideq H\{\;\}\land C$, $K\{\;\}\ideq H\{\;\}\lor C$ or $K\{\;\}\ideq H\{\;\}\ba C$, where $\ba$ is any atom, take
\[
\drB{(H\{A\}\land C)\land\drB{B;[=]-;B\land1};
     [\land\sup\land]-;
     \drB{H\{A\}\land B;|;H\{A\land B\}}
     \land
     \drB{C\land1;[=]-;C}}
\;,\quad
\drB{(H\{A\}\lor C)\land\drB{B;[=]-;B\lor0};
     [\land\sdn\lor]-;
     \drB{H\{A\}\land B;|;H\{A\land B\}}
     \lor
     \drB{C\lor0;[=]-;C}}
\quad\text{and}\quad
\drB{(H\{A\}\ba C)\land\drB{B;[\phi]|;B\ba1};
     [\land\sup\ba]-;
     \drB{H\{A\}\land B;|;H\{A\land B\}}
     \ba
     \drB{C\land1;[=]-;C}}
\;,
\]
respectively, where $\phi$ is as in Lemma~\ref{LemmaConstr}. The remaining cases of $K\{\;\}\ideq C\land H\{\;\}$, $K\{\;\}\ideq C\lor H\{\;\}$ and $K\{\;\}\ideq C\ba H\{\;\}$ can all be dealt with similarly. These derivations contain no instances of the cut or identity rule.

The width of each derivation is $O(n+m+l)$. Note that the height of the third case dominates the other cases; $\phi$ is of height $O(l)$, whereas the corresponding part of the other two derivations is an equality rule of constant height. There are $O(m)$ inductive steps in the derivation, each contributing height $O(l)$, thus the height of the derivation in total is $O(ml)$.
\end{proof}

In decision trees such as $(A\ba B)\ba C$, the truth of the formula does not depend on $B$. The equivalence of that formula to $A\ba C$ is provable within the system. We state and prove the general case below.

\begin{lemma}\label{LemmaDTWEak}
In system\/ $\DTsa$, for every formulae $A$, $B$ and $C$ and every atom $\ba$, there exist cut-free derivations of the form
\[
\drv{K\{A\}\ba  ;|;K\{A\ba B\}\ba C}
\quad\text{and}\quad
\drv{C\ba K\{A\};|;C \ba K\{B\ba A\}}
\;,
\]
for any context $K\{\;\}$. Their width is $O(m+n)$ and their height is $O(m+n)$, where $m$ is the size of $K\{\;\}$ and $n$ is the size of\/ $(A\ba B)\ba C$.
\end{lemma}

\begin{proof}
We construct the first derivation; the second can be obtained as the mirror image of the first. The derivation is shown in Figure~\ref{FigDTWEak}.

We take derivations $\phi$, $\psi$ and $\chi$ to be instances of the constructions that nest and un-nest atom occurrences from Lemma~\ref{LemmaNest}, while $\omega$ is an instance of weakening. The derivation $\pi$ is the construction from Lemma~\ref{LemmaMerge}. The derivation contains no cuts.

The width of the derivation is $O(m+n)$. The subderivations $\phi$ and $\chi$ have height $O(m+n)$, while $\psi$ has height $O(n)$ and $\pi$ has height $O(m)$ so overall the height of the derivation is $O(m+n)$.
\end{proof}

\begin{definition}
We call \emph{DT-weakenings} the derivations constructed in Lemma~\ref{LemmaDTWEak}.
\end{definition}

\begin{figure}[t]
\[
\drv{K\{A\}\ba C;
     [\phi]|;
     \drB{\drB{K\{A\}
               \land
               \drB{\drv{\drB{1;[=]-;1\land1}
                         \ba
                         \drB{0;[=]-;0\land0};
                         [\land\sdn\ba]-;
                         (1\ba0)\land(1\ba0)}};
               [=]-;
               \drB{\drv{K\{A\}\land(1\ba0);
                    [\pi]|;
                    \drB{K\left\{\,
                         \drv{A\land(1\ba0);
                              [=]-;
                              \drB{(A\land(1\ba0))
                                   \lor
                                   \drB{0;[\omega]|;(0\ba1)\land B}};
                              [\psi]|;
                              A\ba B}
                         \,\right\}}}}
               \land
               (1\ba0)}
          \lor
          ((0\ba1)\land C)};
     [\chi]|;
     K\{A\ba B\}\ba C}
\]
\caption{Derivation of DT-weakening.}
\Description{Derivation of DT-weakening.}
\label{FigDTWEak}
\end{figure}

The following definition of the `projection' of an atom on a derivation is key to our proof of cut elimination. The right projection on $\ba$ is defined on formulae by replacing every subformula of the form $A\ba B$ with $B$, with the left projection defined analogously. For example, the right projection on $\ba$ of the formula $(0\lor(0\ba1))\land(1\ba(0\ba1))$ is $(0\lor1)\land1$. The obvious extension of this definition to derivations is well defined, with any inference rules that are broken easily fixed.

\begin{definition}
For any atom $\ba$, we denote the \emph{right projection on} $\ba$ of a derivation $\phi$ in system $\DTsa$ as $\rpr\ba\phi$, which we define inductively as follows.
\begin{enumerate}
\item The base cases are $\rpr\ba0\ideq0$ and $\rpr\ba1\ideq1$.
\item If $\phi$ is a horizontal composition of two derivations $\psi$ and $\chi$ by either the connective $\ba$ or a connective $\beta\not\ideq\ba$ (which is possibly another atom), define
\[
\rpr\ba\left(\,
\drB{A;[\psi]|;B}
\ba
\drB{C;[\chi]|;D}
\,\right)
\ideq
\drB{\rpr\ba C;[\rpr\ba\chi]|;\rpr\ba D}
\quad\text{and}\quad
\rpr\ba\left(\,
\drB{A;[\psi]|;B}
\sbet
\drB{C;[\chi]|;D}
\,\right)
\ideq
\drB{\rpr\ba A;[\rpr\ba\psi]|;\rpr\ba B}
\sbet
\drB{\rpr\ba C;[\rpr\ba\chi]|;\rpr\ba D}
\;.
\]
\item If $\phi$ is the composition of derivations $\psi$ and $\chi$ by the inference rule $\ba\sup\lor$, define
\[
\rpr\ba
\drB{E;
     [\psi]|;
     (A\lor B)\ba(C\land D);
     [\ba\sup\lor]-;
     (A\ba C)\lor(B\ba D);
     [\chi]|;
     F}
\ideq
\drB{\rpr\ba E;
     [\rpr\ba\psi]|;
     \rpr\ba C\land\rpr\ba D;
     [=]-;
     (0\lor\rpr\ba C)\land(0\lor\rpr\ba D);
     [\lor\sdn\land]-;
     (0\land0)\lor(\rpr\ba C\lor\rpr\ba D);
     [=]-;
     \rpr\ba C\lor\rpr\ba D;
     [\rpr\ba\chi]|;
     \rpr\ba F}
\;;
\]
we do analogously when the composition is by inference $\ba\sdn\land$.
\item If $\phi$ is the composition of derivations $\psi$ and $\chi$ by an inference rule $\alpha\sdn\ba$ or $\alpha\sup\ba$, for any connective $\alpha$, define
\[
\rpr\ba
\drB{E;
     [\psi]|;
     (A\salp B)\ba(C\salp D);
     [\alpha\sdn\ba]-;
     (A\ba C)\salp(B\ba D);
     [\chi]|;
     F}
\ideq
\drB{\rpr\ba E;
     [\rpr\ba\psi]|;
     \rpr\ba(C\salp D);
     .;
     \rpr\ba(C\salp D);
     [\rpr\ba\chi]|;
     \rpr\ba F}
\]
and
\[
\rpr\ba
\drB{E;
     [\psi]|;
     (A\ba B)\salp(C\ba D);
     [\alpha\sup\ba]-;
     (A\salp C)\ba(B\salp D);
     [\chi]|;
     F}
\ideq
\drB{\rpr\ba E;
     [\rpr\ba\psi]|;
     \rpr\ba(B\salp D);
     .;
     \rpr\ba(B\salp D);
     [\rpr\ba\chi]|;
     \rpr\ba F}
\;.
\]
\item Otherwise, if $\phi$ is the composition of derivations $\psi$ and $\chi$ by an inference or equality $\rho$, then $\rho$ survives the projection unchanged (although it may act on different formulae) and we define
\[
\rpr\ba
\drB{E;
     [\psi]|;
     A;
     [\rho]-;
     B;
     [\chi]|;
     F}
\ideq
\drB{\rpr\ba E;
     [\rpr\ba\psi]|;
     \rpr\ba A;
     [\rho]-;
     \rpr\ba B;
     [\rpr\ba\chi]|;
     \rpr\ba F}
\;;
\]
if $\rho$ is an equality, note that if $A=B$, then $\rpr\ba A=\rpr\ba B$ for our equivalence relation $=$.
\end{enumerate}
We also analogously define the \emph{left projection on} $\ba$ of a derivation $\phi$, denoted $\lpr\ba\phi$.
\end{definition}

We could see the above definition as a term-rewriting procedure. In that case, note that there are no critical pairs, therefore the procedure is confluent. In other words, given a derivation, a projection is unique.

\newcommand{\pale}[1]{\textcolor{ACMLightBlue}{#1}}
\newcommand{\black}[1]{\textcolor{black!100}{#1}}
\begin{figure}[t]
\[
\drv{1;
     [=]-;
     \drB{\pale{\black1\ba(1\lor0)};
          [\phi]|;
          \drB{\drB{\pale{\black1\ba\pale{0}};
                    [=]-;
                    0\lor\pale{(\black1\ba0)}}
               \mathbin{\pale{\ba}}
               \pale{
               \left(1\lor\drv[border=ACMLightBlue]
                              {\drv{0;[\omega_1]|;1\bb0}}\,\right)}};
          [\pale{\lor\sdn{{\ba}}}]-;
          \pale{(\black0\ba1)\mathbin{\black\lor}
                ((\black1\ba0)\ba(1\bb0))}}
     \land
     \drB{\left(1\lor\drB{0;[\omega_2]|;1\bb0}\,\right)
          \mathbin{\pale{\ba}}
          \pale{(0\lor(1\bb1))};
          [\pale{\lor\sdn\ba}]-;
          \pale{(\black1\ba0)}
          \lor
          \drB{\pale{\black{(1\bb0)}\ba(1\bb1)};
               [\pale{\bb\sdn\ba}]-;
               \drB{\pale{\black1\ba1};
                    [=]-;
                    1}
               \bb
               \pale{(\black0\ba1\pale)}}};
     [\lor\sdn\land]-;
     \drv{\drB{\pale{(\black0\ba1)}\land\pale{(\black1\ba0)};
               [\pale{\mathsf{cut}}]-;
               \pale{\black{(0\land1)}\ba(1\land0)}}
          \lor
          (\pale{((\black1\ba0)\ba(1\bb0))}\lor(1\bb\pale{(\black0\ba1)}));
          [=]-;
          \pale{((\black1\ba0)\ba(1\bb0))}\lor(1\bb\pale{(\black0\ba1)})}}
\]
\caption{An example proof of the implication $((0\ba1)\ba(0\bb1)) \to(1\bb(0\ba1))$, \emph{i.e.}, $\overline{(0\ba1)\ba(0\bb1)}\lor (1\bb(0\ba1))$ in system $\DTsa$. The black elements are those in the \emph{left projection on} $\ba$ of the proof, while the pale elements are those discarded in taking the left projection. The derivation $\phi$ is an instance of DT-weakening, while $\omega_1$ and $\omega_2$ are weakenings. Figure~\ref{FigSmallProof} displays a fragment of this proof, with the definition of $\omega_2$ unfolded.}
\Description{An example proof of the implication $(((0\ba1)\ba(0\bb1)) \to(1\bb(0\ba1)))$, \emph{i.e.}, $(\overline{(0\ba1)\ba(0\bb1)}\lor (1\bb(0\ba1)))$ in system $\DTsa$. The black elements are those in the \emph{left projection on} $\ba$ of the proof, while the pale elements are those discarded in taking the left projection. The derivation $\phi$ is an instance of DT-weakening, while $\omega_1$ and $\omega_2$ are weakenings. Figure~\ref{FigSmallProof} displays a fragment of this proof, with the definition of $\omega_2$ unfolded.}
\label{FigBigProof}
\end{figure}

For an example of the left projection of a proof, see Figure~\ref{FigBigProof}.

\begin{remark}
There are no occurrences of atom $\ba$ in either $\lpr\ba\phi$ or $\rpr\ba\phi$, and thus no cuts on $\ba$ in either – even if there are cuts on $\ba$ in $\phi$.
\end{remark}

\begin{lemma}\label{LemmaReorder}
In system\/ $\DTsa$, for every formula $A$ and every atom\/ $\ba$, there exist derivations
\[
\drv{(\lpr\ba A)\ba(\rpr\ba A);|;A}
\quad\text{and}\quad
\drv{A;|;(\lpr\ba A)\ba(\rpr\ba A)}
\;,
\]
the first of which is cut-free. The size of each derivation is $O(n^3)$, where $n$ is the size of $A$. If $A$ contains no nested atom occurrences $\ba$, the size of each derivation is $O(n^2)$.
\end{lemma}

\begin{proof}
We construct the first derivation; the second can be obtained dually. We proceed by induction on the size of $A$. In the base case, where $A\ideq0$ or $A\ideq1$, we take the derivation
\[
\drv{0\ba0;[=]-;0}
\quad\text{or}\quad
\drv{1\ba1;[=]-;1}
\;,
\]
respectively.

If $A\ideq B\sbet C$, for $\beta\not\ideq\ba$ (but $\beta$ is possibly another atom), we construct
\[
\drv{(\lpr\ba B\sbet\lpr\ba C)\ba
     (\rpr\ba B\sbet\rpr\ba C);
     [\beta\sdn\ba]-;
     \drB{\lpr\ba B\ba\rpr\ba B;
          [\psi]|;
          B}
     \sbet
     \drB{\lpr\ba C\ba\rpr\ba C ;
          [\chi]|;
          C}}\;,
\]
where $\psi$ and $\chi$ are obtained via the inductive assumption.

If $A\ideq B\ba C$, we construct
\[
\drv{\lpr\ba B\ba\rpr\ba C;
     [\omega]|;
     \drB{\drB{\lpr\ba B\ba\rpr\ba B;
               [\psi]|;
               B}
          \ba
          \drB{\lpr\ba C\ba\rpr\ba C;
               [\chi]|;
               C}}}\;,
\]
where $\psi$ and $\chi$ are obtained via the inductive assumption and $\omega$ is made up of two instances of DT-weakening. Note that if $A$ contains no nested atom occurrences $\ba$, then $A\ideq B\ba C\ideq(\lpr\ba B)\ba(\lpr\ba C)\ideq(\lpr\ba A)\ba(\rpr\ba A)$, \emph{i.e.}, we can simply adopt the trivial derivation.

The width of the derivation is $O(n)$; the DT-weakenings used in the second case are of height $O(n)$, while the $\beta\sdn\ba$ of the first case is of constant size. Thus, in the worst case, we have DT-weakenings of height $O(n)$ stacked $O(n)$ times by the inductive step, for a total height of $O(n^2)$. Thus, the total size is $O(n^3)$. However, if $A$ contains no nested atom occurrences $\ba$, we never need the DT-weakening case and thus we have constant size rules stacked $O(n)$ times, for a total size of $O(n^2)$.
\end{proof}

The previous lemma shows the semantic equivalence between $(\lpr\ba A)\ba(\rpr\ba A)$ and $A$ because it establishes the existence of derivations that go from one formula to the other and back. (The same could be established semantically, of course.)

\begin{theorem}
In system\/ $\DTsa$, the cut rule is admissible.
\end{theorem}

\begin{proof}
Take a proof $\phi$ of a formula $A$. If there exists a cut on $\ba$, construct
\[
\drv{1;
     [=]-;
     \drB{\drv[ACMBlue]{1;
                        [\lpr\ba\phi]|;
                        \lpr\ba A}
          \ba
          \drv[ACMRed]{1;
                       [\rpr\ba\phi]|;
                       \rpr\ba A}};
     [\chi]|;
     A}\;,
\]
where $\chi$ is the cut-free derivation of Lemma~\ref{LemmaReorder}. This proof has no cuts on $\ba$. Iterating this construction yields a cut-free proof.
\end{proof}

\begin{figure}[t]
\[
\drv{1;
     [=]-;
     \drv[ACMBlue]{(0\lor1)
                   \land
                   \left(1\lor\drB{0;[\omega_2]|;1\bb0}\,\right);
                   [\lor\sdn\land]-;
                   (0\land1)\lor(1\lor(1\bb0));
                   [=]-;
                   1\lor(1\bb0)}
     \ba
     \drv[ACMRed]{\left(1\lor\drB{0;[\omega_1]|;1\bb0}\,\right)
                  \land
                  (0\lor(1\bb1));
                  [\lor\sdn\land]-;
                  (1\land0)\lor((1\bb0)\lor(1\bb1));
                  [=]-;
                  (1\bb0)\lor(1\bb1)};
     [\lor\sdn\ba]-;
     \drB{1\ba(1\bb0);
          [\psi]|;
          (1\ba0)\ba(1\bb0)}
     \lor
     \drB{(1\bb0)\ba(1\bb1);
          [\bb\sdn\ba]-;
          \drB{1\ba1;[=]-;1}\bb(0\ba1)}}
\]
\caption{An example of the cut-elimination construction applied once, for atom $\ba$, to the proof we presented in Figure~\ref{FigBigProof}. The two shaded boxes at the left and the right of $\ba$ denote the left and right projections of the proof of Figure~\ref{FigBigProof}, respectively. The derivation $\psi$ is an instance of DT-weakening. Some equality rule instances have been grouped.}
\Description{An example of the cut-elimination construction applied once, for atom $\ba$, to the proof we presented in Figure~\ref{FigBigProof}. The two shaded boxes at the left and the right of $\ba$ denote the left and right projections of the proof of Figure~\ref{FigBigProof}, respectively. The derivation $\psi$ is an instance of DT-weakening. Some equality rule instances have been grouped.}
\label{FigCutElExample}
\end{figure}

For an example of the cut-elimination procedure, see Figure~\ref{FigCutElExample}. Note that the procedure blows up the size of the proof exponentially. In particular, it grows exponentially in the number of \emph{distinct} atoms in the original proof upon which there are cuts (rather than exponentially in the number of cuts). The procedure is not confluent, due to an arbitrary ordering of the atoms upon which there are cuts. The proof is, somewhat, conceptually similar to the `experiments method' for system $\SKS$ \cite{Ralp:19:Modular-:yq}, which is, however, confluent.

The semantical nature of the cut-elimination procedure, in addition to its simplicity and the simplicity of the proof system itself, are the reasons why we consider $\DTsa$ natural.

\section{Statman Tautologies}\label{SectStatman}

We mentioned in the introduction that our proof system generates efficient proofs. Indeed, we show that a certain class of tautologies due to Statman, which cannot have better than exponential cut-free proofs in the sequent calculus, have polynomial cut-free proofs in our system. We achieve this by using projections – the same construction that we use for cut elimination. Further, the proofs are semantically natural, being a formalisation of the argument described in the introduction.

The emphasis of this section is on the method rather than the result. We already knew that Statman tautologies can be proved with no cuts and in polynomial time in deep inference \cite{BrusGugl:07:On-the-P:fk}. Indeed, a very small amount of deep inference is necessary to achieve that, meaning that allowing inference just below the root connective of formulae is all that is needed to obtain polynomial-size proofs. We also know that that small amount of deep inference is enough to polynomially simulate proofs with inferences at any depth \cite{Das:11:On-the-P:fk} and still retain cut-freeness. Intuitively, we achieve those speed-ups because deep inference captures the distributivity laws more efficiently than the sequent calculus, without involving cuts.

Cuts provide an essentially different compression mechanism from distributivity. Especially when they are reduced to their atomic form, cuts can be interpreted as a case-analysis mechanism. In fact, in an atomic cut, the two dual hypotheses on the truth value of an atom are formulated and then employed in the proof above the cut. The point made in this section is that polynomial-size proofs of the Statman tautologies can be obtained via a natural case analysis, as explained in the introduction, but without cuts.

\begin{figure}[t]
\[
\drv{1;
     [=]-;
     \drB{\drB{(1\lor0)\bb_1(0\lor1);
               [\lor\sdn\bb_1]-;
               \drB{(1\bb_1 0);
                   [=]-;
                   1\land\bar b_1}
               \lor
               \drB{(0\bb_1 1);
                   [=]-;
                   0\lor b_1}}
          \ba_1
          \left(\left(0
                      \land
                      \drB{0;
                           [\omega_1]|;
                           \bar b_1}\,
                \right)
                \lor
                \left(1
                      \lor
                      \drB{0;
                           [\omega_2]|;
                           b_1}\,
                \right)
          \right)};
     [\phi]|;
     (\bar a_1\land\bar b_1)\lor(a_1\lor b_1)}
\]
\caption{Base case for the construction of polynomial-size proofs of Statman tautologies.}
\Description{Base case for the construction of polynomial-size proofs of Statman tautologies.}
\label{FigStatmanBase}
\end{figure}

\begin{figure}[t]
\[
\begin{array}{@{}l@{}}
\drv{S_{n-1};
     [\pi]|;
     \drB{\drB{S_{n-1};
               [\rho]|;
               \drB{\drB{1;
                         [=]-;
                         \phi}
                    \bb_n
                    \drB{S_{n-1};
                         [=]-;
                         \begin{array}{@{}r@{}l@{}}
                            0\lor{}                                    &
                            (1\land\bar a_{n-1}\land1\land\bar b_{n-1}) \\
                         \noalign{\smallskip}
                            {}\lor\bigvee_{k=1}^{n-2}{}                &
                            (1\land A^{n-1}_k\land1\land B^{n-1}_k)
                            \lor(a_1\lor b_1)                           \\
                         \end{array}}};
               [\chi]|;
               \drB{\begin{array}{@{}r@{}l@{}}
                       \drB{\bar b_n;
                            [=]-;
                            1\land\bar b_n}\lor{}        &
                       \left(\,\drB{b_n;
                                      [=]-;
                                      0\lor b_n}
                                 \land\bar a_{n-1}\land
                                 \drB{b_n;
                                      [=]-;
                                      0\lor b_n}
                                 \land\bar b_{n-1}
                       \right)                            \\
                    \noalign{\medskip}
                       {}\lor\bigvee_{k=1}^{n-2}{}       &
                       \left(\,\drB{b_n;
                                    [=]-;
                                    0\lor b_n}
                               \land A^{n-1}_k\land
                               \drB{b_n;
                                    [=]-;
                                    0\lor b_n}
                               \land B^{n-1}_k
                       \right)
                       \lor(a_1\lor b_1)                  \\
                    \end{array};
                    .;
                    \lpr{\ba_n}S_n}}
          \ba_n
          \drB{S_{n-1};[\psi]|;\rpr{\ba_n}S_n}};
     [\theta]|;
     S_n}
\\\noalign{\bigskip\bigskip}
\text{where}\quad\phi\ideq1\lor
\drB{0;
     [\omega]|;
     \begin{array}{@{}r@{}l@{}}
                                                   &
        (0\land\bar a_{n-1}\land0\land\bar b_{n-1}) \\
     \noalign{\smallskip}
        {}\lor\bigvee_{k=1}^{n-2}{}                &
        (0\land A^{n-1}_k\land0\land B^{n-1}_k)
        \lor(a_1\lor b_1)                           \\
     \end{array}}
\\\noalign{\bigskip\bigskip}
\text{and}\quad\psi\ideq
\drB{(\bar a_{n-1}\land\bar b_{n-1})
     \lor\bigvee_{k=1}^{n-2}(A^{n-1}_k\land B^{n-1}_k)
     \lor(a_1\lor b_1);
     [=]-;
     \begin{array}{@{}r@{}l@{}}
        \left(0\land\drB{0;
                         |;
                         \bar b_n}\,
        \right)\lor{}                &
        \left(\left(1\lor\drB{0;
                              |;
                              b_n}\,
              \right)
              \land\bar a_{n-1}\land
              \left(1\lor\drB{0;
                              |;
                              b_n}\,
              \right)
              \land\bar b_{n-1}
        \right)                       \\
     \noalign{\medskip}
        {}\lor\bigvee_{k=1}^{n-2}{}  &
        \left(\left(1\lor\drB{0;
                              |;
                              b_n}\,
              \right)
              \land A^{n-1}_k\land
              \left(1\lor\drB{0;
                              |;
                              b_n}\,
              \right)
              \land B^{n-1}_k
        \right)
        \lor(a_1\lor b_1)
     \end{array}}                     \\
\end{array}
\]
\caption{Inductive case for the construction of polynomial-size proofs of Statman tautologies.}
\Description{Inductive case for the construction of polynomial-size proofs of Statman tautologies.}
\label{FigStatmanInduct}
\end{figure}

\begin{definition}
We call \emph{Statman tautologies} the formulae $S_1$, $S_2$, \dots:
\[
\renewcommand\arraystretch{1.5}
\begin{array}{@{}r@{}l@{}l}
S_1&{}\ideq{}&(\bar a_1\land\bar b_1)\lor(a_1\lor b_1)\;,            \\
   &\hbox to0pt{$\;\cdots\;$,\hss}                                   \\
S_n&{}\ideq{}&(\bar a_n\land\bar b_n)\lor((A^n_{n-1}\land B^n_{n-1})
                           \lor\cdots\lor (A^n_1    \land B^n_1    ))
                                     \lor(a_1\lor b_1)\;,            \\
   &\hbox to0pt{$\;\cdots\;$,\hss}                                   \\
\end{array}
\]
 where $a_i$ and $b_i$ stand for $0\ba_i 1$ and $0\bb_i 1$, and $\bar a_i$ and $\bar b_i$ stand for $1\ba_i 0$ and $1\bb_i 0$, we work modulo associativity and, for $n>k\geq1$:
\[
\renewcommand\arraystretch{1.5}
\begin{array}{@{}l@{}}
A^n_k\ideq(a_n\lor b_n)\land\cdots\land(a_{k+1}\lor b_{k+1})\land\bar a_k\;,\\
B^n_k\ideq(a_n\lor b_n)\land\cdots\land(a_{k+1}\lor b_{k+1})\land\bar b_k\;.\\
\end{array}
\]
\end{definition}

Note that, in the previous definition, $A^n_k\ideq(a_n\lor b_n)\land A^{n-1}_k$ and $B^n_k\ideq(a_n\lor b_n)\land B^{n-1}_k$, whenever $n-1>k$.

\begin{theorem}
There exist cut-free proofs of the Statman tautologies of size $O(m^{2.5})$, where $m$ is the size of the tautologies.
\end{theorem}

\begin{proof}
Given $S_1$, $S_2$, \dots, $S_n$, where $S_n$ is of size $O(m)=O(n^2)$, the proof is by induction on $n$; the base case is in Figure~\ref{FigStatmanBase} and the inductive case in Figure~\ref{FigStatmanInduct}.

In Figure~\ref{FigStatmanBase}, $\phi$ is the construction from Lemma~\ref{LemmaReorder}, $\omega_1$ and $\omega_2$ are weakenings, and the size of the whole derivation is (obviously) a constant.

In Figure~\ref{FigStatmanInduct}, each derivation from $S_i$ to $S_{i+1}$ is of size $O(m^2)$ and again we use projections: $\theta$ and $\chi$ are the construction from Lemma~\ref{LemmaReorder}, with size $O(m^2)$ (since there are no nested $\ba_n$'s in $S_n$); $\omega$ is a weakening with size $O(m)$; $\pi$ is a cocontraction of size $O(m^2)$; $\rho$ is the construction from Lemma~\ref{LemmaConstr}, also with size $O(m^2)$. There are $O(n)=O(m^{0.5})$ inductive steps in the proof, giving a total size of $O(m^{2.5})$.
\end{proof}

This bound on the size of Statman tautologies matches the bound for dag-like, cut-free sequent proofs presented in \cite{ClotKran:02:Boolean-:xr} (Theorem~5.3.3).

The example on Statman tautologies that we have presented here shows that we can have proofs that are short, based on natural case analysis and cut-free. No other proof system that we know of exhibits those three properties. At this point, we do not know how far our case-analysis technique can go, but we do not expect it to have the same power as case analysis by cuts. This could be the subject of future research.

\section{Strict Decision Trees}

We have seen how to translate between $\DT$ and $\Prop$ within system $\DTsa$. Here, we see how to translate between $\DT$ and $\SDT$, and thus also between $\Prop$ and $\SDT$.

\begin{lemma}\label{LemmaDTConv}
In system\/ $\DTsa$, for every formula $A$ in\/ $\DT$, there exists some semantically equivalent formula $B$ in\/ $\SDT$ and derivations
\[
\drv{A;|;B}
\quad\text{and}\quad
\drv{B;|;A}
\;.
\]
\end{lemma}
\begin{proof}
Here, we construct the second derivation. The first derivation can be obtained dually.

We proceed by induction on the number of distinct atoms (not atom occurrences) appearing in $A$. If there are no atoms in $A$, the formula is made up of only conjunctions and disjunctions of units, which is equal under $=$ to a unit $0$ or $1$ and thus is in $\SDT$.

If $A$ contains at least one atom $\ba$, we construct the following derivation:
\[
\drv{\drB{\drB{C;
               [\phi]|;
               \lpr\ba A}
          \ba
          \drB{D;
               [\psi]|;
               \rpr\ba A}};
     [\chi]|;
     A}
\;,
\]
where $\chi$ is the construction from Lemma~\ref{LemmaReorder}, and $\phi$ and $\psi$ can be obtained from the inductive hypothesis since both $\lpr\ba A$ and $\rpr\ba A$ must contain strictly fewer distinct atoms than $A$ (they no longer contain any occurrences of $\ba$). Take $B\ideq C\ba D$. Because $C$ and $D$ are both in $\SDT$ by the inductive hypothesis, $B$ is also in $\SDT$. That $B$ is semantically equivalent to $A$ follows from Lemma~\ref{LemmaReorder}.
\end{proof}

\begin{remark}\label{RemOrdered}
Lemma~\ref{LemmaDTConv} can be used to obtain a $B$ that is an \emph{ordered} decision tree, \emph{i.e.}, given a total ordering $\leq$ on atoms, in every path from the root to a leaf of (the syntax tree of) $B$, each atom occurs at most once and the sequence of atom occurrences given by each such path respects $\leq$.
\end{remark}

\begin{remark}
Lemma~\ref{LemmaDTConv} yields an easy alternate proof of completeness. Given a tautology $A$, the construction of Lemma~\ref{LemmaDTConv} yields an equivalent decision tree $B$ and a derivation $\phi$ from $B$ to $A$. In fact, by Remark~\ref{RemOrdered}, we can take $B$ to be an ordered decision tree. Thus, since $B$ is also a tautology, $B$ must be equal under $=$ to the unit $1$, and so $\phi$ is a proof – in fact, a cut-free one.
\end{remark}

We can now translate between the two languages $\Prop$ and $\SDT$ within our system.

\begin{theorem}
In system\/ $\DTsa$, given a formula $A$ in\/ $\Prop$, there exists a semantically equivalent formula $B$ in\/ $\SDT$, a derivation from $A$ to $B$ and a derivation from $B$ to $A$. Given a formula $C$ in\/ $\SDT$, there exists a semantically equivalent formula $D$ in\/ $\Prop$, a derivation from $C$ to $D$ and a derivation from $D$ to $C$.
\end{theorem}

\begin{proof}
It follow immediately from Lemmas~\ref{LemmaEquiv} and \ref{LemmaDTConv}.
\end{proof}

A \emph{reduced ordered decision tree (RODT)} is an ordered decision tree where every subformula of the form $(A\ba A)$ is replaced with $A$ \cite{Wege:00:Branchin:zr}. For a given Boolean function and a total ordering on atoms, there is a unique RODT that represents that function and respects the ordering in the sense of Remark~\ref{RemOrdered}. Instead of decision trees, RODTs are often used in practical applications.

\begin{remark}
Our proof system contains derivations that correspond to runs of some common algorithms used for decision trees. An especially common algorithm, known as `apply', takes two RODTs $A$ and $B$ with the same ordering and a binary Boolean connective $\land$ or $\lor$, and produces a new RODT with that ordering, which represents the same function as $A\land B$ or $A\lor B$ respectively.
\end{remark}

\begin{figure}[t]
\[
\drv{(((0\ba1)\bb(1\ba0))\bc(0\ba1))\land((0\ba1)\bc(1\ba0));
     [\land\sup{\bc}]-;
     \drB{((0\ba1)\bb(1\ba0))\land
          \drB{\drB{0;
                    [=]-;
                    0\bb0}
               \ba
               \drB{1;
                    [=]-;
                    1\bb1};
               [\ba\sup\bb]-;
               (0\ba1)\bb(0\ba1)};
          [\land\sup\bb]-;
          \drB{(0\ba1)\land(0\ba1);
               [\land\sup\ba]-;
               \drB{0\land0;[=]-;0}
               \ba
               \drB{1\land1;[=]-;1}}
          \bb
          \drB{(1\ba0)\land(0\ba1);
               [\land\sup\ba]-;
               (1\land0)\ba(0\land1);
               [=]-;
               0}}
     \bc
     \drB{(0\ba1)\land(1\ba0);
          [\land\sup\ba]-;
          (0\land1)\ba(1\land0);
          [=]-;
          0}}
\]
\caption{A derivation corresponding to a run of the common `apply' algorithm on two reduced ordered DTs.}
\Description{A derivation corresponding to a run of the common `apply' algorithm on two reduced ordered DTs.}
\label{FigApply}
\end{figure}

For example, in the case of applying $\land$ to two RODTs $((0\ba1)\bb(1\ba0))\bc(0\ba1)$ and $(0\ba1)\bc(1\ba0)$, which both have ordering $\bc<\bb<\ba$, the derivation corresponding to a run of the algorithm is given in Figure~\ref{FigApply}. The conclusion of the derivation is the RODT representation of the premiss, respecting the ordering $\bc<\bb<\ba$.

\section{Conclusions and Further Research}

Decision trees and closely related structures have been extended in the past with Boolean operators, typically to improve their expressiveness and efficiency, as in \cite{AndeHulg:97:Boolean-:gd}. More recently, proof theory has entered the picture in \cite{BussDasKnop:20:Proof-Co:eu}, where the authors provide sequent proof systems for decision trees and study their complexity. Against that background, the present work could be seen as a way to design a proof system for decision trees whose proof theory is guaranteed to be natural by the previous research on deep inference. Indeed, nested atoms and the corresponding rules could simply be added to the current standard system for first-order classical logic, which is presented in \cite{Brun:06:Cut-Elim:cq}, and that would effortlessly extend the expressivity of our proof system. We believe that our cut-elimination technique based on projections would extend to that case but we have not explored this yet.

The simplicity of our proof system and cut-elimination procedure is striking, and is a direct consequence of deep inference's ability to deal with non-commutative connectives \cite{Tiu:06:A-System:ai}. Also, the work \cite{Rove:18:Subatomi:zp} is very close to ours because it uses subatomic logic in deep inference to study the normalisation of a proof system with projections, which behave very similarly to decision trees. We would expect that, eventually, all these lines of research will converge towards a standard proof theory of Boolean functions, including decision trees, that supports the applications and whose complexity is well understood.

The expressiveness of our proof system can naturally be extended by quantifiers and modalities. There is now a rich literature in deep inference about standard methods to do so \cite{Gugl::Deep-Inf:uq}. As we mentioned, the cut-elimination procedure in this paper is exponential in the size of the proof to be normalised. However, in deep inference, quasi-polynomial cut-elimination procedures are available for propositional classical logic \cite{BrusGuglGundPari:09:Quasipol:kx,Jera::On-the-C:kx}. It is natural to wonder whether those results could be extended to our proof system, and possibly improved.

\begin{acks}
We would like to thank Andrea Aler Tubella, Victoria Barrett, Anupam Das and Willem Heijltjes for helpful exchanges and insights. We also thank the referees for criticism and suggestions that have greatly benefited the paper. This research has been supported by \grantsponsor{EP/K018868/1}{EPSRC}{https://gow.epsrc.ukri.org/NGBOViewGrant.aspx?GrantRef=EP/K018868/1} project \grantnum{EP/K018868/1}{EP/K018868/1}, `Efficient and Natural Proof Systems'.
\end{acks}

\bibliographystyle{ACM-Reference-Format}
\bibliography{biblio}

\end{document}